\documentclass[12pt,reqno,tbtags]{amsart}

\usepackage{setspace}

\usepackage{amsmath, amsfonts, amsthm}

\usepackage{mathtools}
\usepackage{color}
\usepackage{float}
\usepackage{url}
\usepackage{comment}

\usepackage{pgfplots}
\pgfplotsset{compat=1.18}

\usepackage[dvipsnames]{xcolor}

\usepackage[top=3cm, bottom=2.5cm, left=1cm, right=1cm]{geometry}

\usepackage{newcent}
\usepackage[T1]{fontenc}
\usepackage[utf8]{inputenc}
\usepackage{babel}
\usepackage{bbm}
\usepackage{xspace}
\usepackage[hidelinks]{hyperref}
\usepackage{enumitem}
\usepackage[most]{tcolorbox}
\usepackage{xpatch}
\usepackage{stmaryrd}
\usepackage{amsmath, amssymb, amsthm, amsfonts}

\usepackage{tikz}
\usepackage{pgfplots}
\pgfplotsset{compat=1.18}

\usepackage{amssymb}
\usepackage{enumitem}
\usepackage{makecell}
\usepackage{graphicx}%
\usepackage{multirow}%
\usepackage{amsmath,amssymb,amsfonts}%

\usepackage{mathrsfs}%
\usepackage[title]{appendix}%
\usepackage{xcolor}%
\usepackage{textcomp}%
\usepackage{manyfoot}%
\usepackage{booktabs}%
\usepackage{algorithm}%
\usepackage{algorithmicx}%
\usepackage{algpseudocode}%
\usepackage{algorithm}
\usepackage{algpseudocode}

\algdef{SE}[SWITCH]{Switch}{EndSwitch}[1]{\textbf{switch} #1 \textbf{do}}{\textbf{end switch}}
\algdef{SE}[CASE]{Case}{EndCase}[1]{\textbf{case} #1:}{\textbf{end case}}
\algdef{SE}[DEFAULT]{Default}{EndDefault}{\textbf{default}:}{\textbf{end default}}
\algtext*{EndCase}
\algtext*{EndDefault}
\usepackage{listings}

\usepackage{mathtools}

\usepackage{comment}
\usepackage{caption}
\usepackage{subcaption}

\usepackage{fullpage}

\usepackage{fancyhdr}
\newcommand{\I}{\mathbbm{1}}
\newcommand{\EXP}{\mathbb{E}}

\newcommand{\PROB}{\mathbb{P}}

\newcommand{\inlaw}{\buildrel {\mathcal L} \over =}
\newcommand{\RR}{\mathbb{R}}

\newcommand{\isdef}{\buildrel {\rm def} \over =}

\usepackage{makecell} 

\newtheorem{theorem}{Theorem}

\newtheorem{remark}{Remark}

\usepackage[numbers, sort&compress]{natbib}

\usepackage{ebgaramond}

\begin{document}

\title[PearsonIV]{The acceptance-complement method revisited}

\author{Luc Devroye$\dagger$}
\thanks{$\dagger$School of Computer Science, McGill University, 
		Montr\'eal, Qu\'ebec,  Canada: {lucdevroye@gmail.com}. Supported by the Natural Sciences and Engineering Research Council of Canada (NSERC) under grant number RGPIN-2024-04164}


\begin{abstract}
We revisit the acceptance-complement method in random variate generation and show how it can replace the rejection method in many examples.  While the rejection method has geometrically distributed execution times, the acceptance-complement method has a constant (deterministic) run time and qualifies as a ``one-liner''.  
We show how this method can be used to efficiently generate random variates from several distributions, such as the gamma and beta. In addition, we show that there is an acceptance-complement method that is valid for all log-concave densities with known location of the mode and black-box type access to the density.  
\end{abstract}


\keywords{Random variate generation, 
Acceptance-complement method, 
Rejection method, 
Simulation, Monte Carlo method, 
Gamma distribution,
Beta distribution,
Log-concave distributions,
Probability inequalities}

\subjclass[2010]{65C10, 65C05, 11K45, 68U20}

\maketitle




\section{The acceptance-complement method}\label{AC}

To generate random variates with a given density $f$ on $\RR$,
Kronmal and Peterson \cite{Kronmal1, Kronmal2} (see also \cite{ahrens1982computer, deak1981economical}) introduced
the acceptance-complement method, which requires two nonnegative functions, $q$ and $r$, such that
\begin{itemize}
\item [(i)] $r$ is proportional to a density for which we have a simple random variate generator at hand, and $f \ge r$.
\item [(ii)] $q$ is a density for which we have a simple random variate generator at hand, and $q \ge f-r$.
\end{itemize}
We assume that we have a source that can generate independent, identically distributed (i.i.d.) uniform $[0,1]$ random variates of infinite precision.

\begin{algorithm}[H]
\caption{The generic acceptance-complement method}\label{AC1}
\begin{algorithmic}[1]
\State let $Y$ have density $q$ 
\State let $U$ be uniform on $[0,1]$
\If {$U q(Y) \le f(Y)-r(Y)$}
\State return $X \gets Y$ and halt 
\Else 
\State return a random variate $X$ with density proportional to $r$ 
\EndIf
\Comment{$X$ has density $f$}
\end{algorithmic}
\end{algorithm}

In contrast, if $f \le cq$, where $q$ is a density and $c \ge 1$ is a constant (the so-called rejection constant), then von Neumann's rejection method \cite{vonneumann1963various, D5} is as follows:

\begin{algorithm}[H]
\caption{The rejection method}\label{R}
\begin{algorithmic}[2]
\Repeat
\State let $Y$ have density $q$ 
\State let $U$ be uniform on $[0,1]$
\Until {$U cq(Y) \le f(Y)$}
\State return $X \gets Y$ and halt
\Comment{$X$ has density $f$}
\end{algorithmic}
\end{algorithm}

The number of iterations before halting is geometrically distributed with parameter $1/c$, and the expected number of iterations is $c$.  When $n$ i.i.d. random variates are generated 
using the rejection method, the worst-case number of iterations for one variate is in probability asymptotic to $$\log_{1/(1-1/c)} n.$$
In contrast, the acceptance-complement method requires only one iteration. 

The second advantage of the acceptance-complement method is that it provides a compact representation of a random variate $X$ with density $f$ as a fixed explicit function of a finite number of i.i.d.\ uniform $[0,1]$ random variates, assuming that random variates from both $q$ and $r$ have such representations.
In \cite{D1}, these are called one-liners, since one can generate such random variables using only one line of code.
Not having a loop has the added benefit that, if necessary, one can hard-code or machine-code such algorithms quite efficiently.
Well-known examples include
several distributions with explicitly invertible distribution functions, and famous cases in which a random variate can be written as a function of two or more independent uniform $[0,1]$ random variates. In the list below, $U, U_1, U_2, \ldots$ are independent uniform $[0,1]$ random variates. 
\begin{itemize}
    \item[(i)] The exponential random variable: $\log\left(\frac 1U\right)$.
    \item[(ii)] The Maxwell random variable: $\sqrt{2 \log\left(\frac 1U\right)}$.
    \item[(iii)] The Gumbel random variable: $\log\log \left(\frac 1U\right)$.
     \item[(iv)] The logistic or sigmoid random variable: $\log \left(\frac {U}{1-U}\right)$.
     \item[(v)] The Weibull random variable with parameter $a > 0$: $\log^a \left(\frac 1U\right)$.
     \item[(vi)] The Student $t_1$  or Cauchy random variable: $\tan (\pi(U-1/2))$.
     \item[(vii)] The Student $t_2$  random variable: $\frac{2U-1}{\sqrt{2U(1-U)}}$ \cite{LL}.
    \item[(viii)] The standard Gaussian random variable: $\sqrt{2 \log(1/U_1)} \cos ( 2 \pi U_2 )$. This is known as the Box-M\"uller or polar method.
    \item[(ix)]
    The unilateral stable random variable  $S_\alpha$ of parameter $\alpha \in (0,1)$ can be generated as $
\left( { A(\pi U_1) \over \log(1/U_2) } \right)^{1-\alpha \over \alpha}$, where 
$$
A(u) \isdef \left\{  {( \sin (\alpha u))^\alpha (\sin ((1-\alpha)u)^{1-\alpha}
                     \over
		     \sin u }
		     \right\}^{1 \over 1-\alpha}.
$$
This method is due to Kanter \cite{kanter1975stable}.
 \item[(x)] The general stable distribution also has a one-liner based upon a function of two uniform $[0,1]$ random variables based on an integral representation
of Zolotarev \cite{zolotarev1966representation, zolotarev1986one} and developed by Chambers, Mallows and Stuck \cite{chambers1976method}.
\item[(xi)] The Student $t_a$ random variable with parameter  $a > 0$:  $\sqrt{a} \sin(2 \pi U_1 ) \sqrt{U_2^{-2/a} - 1}$. This is known as Bailey's polar method \cite{Bailey}.
\end{itemize}
We can include various functions of random variables with known exact forms to this extensive list. However, it long remained unclear whether beta and gamma distributions (along with their descendants) could be added, provided one limits the operations to standard functions such as the exponential, logarithm, and trigonometric functions, excluding the beta or gamma functions themselves \cite{D1}. Greaves \cite{greaves2026, greaves2026extended} answered the latter question in the affirmative.  In this paper, we give a simpler method than that developed by Greaves, and add a generic acceptance-complement algorithm that works for all log-concave densities.

The algorithms given below are loopless---execution of the programs proceeds linearly, perhaps interrupted by branches, also known as ``if'' statements or switches.  Such algorithms are equivalent to one-liners if one accepts the indicator function as a member of the permitted family of functions.

The design of a rejection method centers around finding a tight envelope to keep the constant $c$ small.  
In contrast, the acceptance-complement method benefits from having a good lower bound $f \ge r$, so that one has more freedom to find at least one good upper bound $q$ for the difference $f-r$. The restriction is that $q$ must be a density, or, as we shall see below, a positive function with $\int q \le 1$.  Thus, the design starts with the lower bound $r$.

\section{The modified acceptance-complement method}\label{AC2}

To generate random variates with a given density $f$ on $\RR$, one could relax the requirements as follows:
\begin{itemize}
\item [(i)] $r$ is proportional to a density for which we have a simple random variate generator at hand, and $f \ge r$.
\item [(ii)] $q$ is a nonnegative function with $\int q \le 1$ such that $q \ge f-r$, and we can generate random variates with density proportional to $q$.
\end{itemize}

\begin{algorithm}[H]
\caption{The modified acceptance-complement method}\label{ACM}
\begin{algorithmic}[3]
\State let $V$ be uniform on $[0,1]$
\If {$V \ge 1 - \int q$}
\State let $Y$ have density proportional to $q$ 
\State let $U$ be uniform on $[0,1]$
\If {$U q(Y) \le f(Y)-r(Y)$}
\State return $X \gets Y$ and halt 
\EndIf 
\EndIf
\State return a random variate $X$ with density proportional to $r$ 
\Comment{$X$ has density $f$}
\end{algorithmic}
\end{algorithm}

The design of an algorithm requires a squeeze of $f$ such that
$$
0 \le r \le f \le q+r,
$$
where both $q$ and $r$ have integrals less than or equal to one.

\medskip

\section{Example: the gamma density}\label{Gamma}

The gamma density with parameter $a > 0$ is given by
$$
f(x) =\frac{x^{a-1} e^{-x}}{\Gamma (a)}, x > 0,
$$
where $\Gamma$ is Euler's gamma function:
$$
\Gamma (a) = \int_0^\infty x^{a-1} e^{-x} \, dx.
$$
For $a \ge 1$, the gamma density is log-concave, hence unimodal,  with a unique mode at $a-1$. The density itself is concave
if
$$ | x - (a-1) | \le \sqrt{a-1}, x \ge 0,
$$
and convex otherwise.
Set
$$
x_- = \max(0, a-1 - \sqrt{a-1},
x_+ = a-1 + \sqrt{a-1}.
$$
A trivial, yet surprisingly useful, lower bound is
$$
f(x) \ge r(x) =
\begin{cases}
0 & x < x_- \\
  f(x_-)& x_- \le x \le a-1, \\ 
  f(x_+)& a-1  \le x \le x_+ \\
  0& x > x_+.
\end{cases}
$$
The case $x_- = 0$ only occurs when $1 \le a \le 2$,
so we assume in the sequel that $a \ge 2$.
A random variate with density proportional to $r$ is simply generated as
$$
a-1 + \sqrt{a-1} \left( 2 \I_{U \le \frac{f(x_+)}{f(x_+) + f(x_-)} } -1 \right).
$$
With this in hand, we are done if we can find a density $q$ with a simple representation in terms of uniform $[0,1]$ random variables that dominates $f-r$.
To this end, we use the log-concavity of $f$ and note that the first derivative of $\log f$ is
$$
\frac{a-1-x}{x}.
$$
At $x_+$, this derivative is $-\sqrt{a-1}/x_+$, while at $x_-$, it is $\sqrt{a-1}/x_-$ when $x_- > 0$ and $\infty$ when $x_-=0$.
Thus,
\begin{align}
\label{gammatails}
f(x) \le
\begin{cases}
    f(x_+) e^{-(x-x_+)\frac{\sqrt{a-1}}{x_+} }& x > x_+, \\
    f(a-1) & x_- \le x \le x_+, \\
    f(x_-) e^{-(x_- -x)\frac{\sqrt{a-1}}{x_-} }& 0 < x < x_-, \\
    0 & x \le 0.
\end{cases}
\end{align}
A random variable with density proportional to the first upper bound (the right tail) is
$$
a-1 + \sqrt{a-1} + \frac {x_+ E}{\sqrt{a-1}}
\inlaw
a-1 + \sqrt{a-1}\left(1 + \log\left(\frac 1U\right) \left(1 + \frac{1}{\sqrt{a-1}}\right) \right),
$$
where $E$ is exponential, and $U$ is uniform on $[0,1]$. 
A random variable with density proportional to the second upper bound (the left tail) when $x_- > 0$ is
$$
a-1 - \sqrt{a-1} - \frac {x_- E}{\sqrt{a-1}}
\inlaw
a-1 - \sqrt{a-1}\left(1 + \log\left(\frac 1U\right) \left(1 - \frac{1}{\sqrt{a-1}}\right) \right),
$$
where $E$ is exponential, and $U$ is uniform on $[0,1]$. 
The designer of the algorithm also needs the areas under these envelopes. For the right tail, it is
$$
\frac {f(x_+) x_+}{\sqrt{a-1}},
$$
and for the left tail, we get
$$
\frac {f(x_-) x_-}{\sqrt{a-1}}.
$$
Our first candidate for $q$ is the function defined by the upper bounds for the tails given in \eqref{gammatails}.
On $[x_-, a-1]$ and $[a-1,x_+]$, we can upper-bound $f(x)-r(x)$ by $f(a-1)-f(x_-)$ and $f(a-1)-f(x_+)$, respectively. Thus, we define
$$
q(x) =
\begin{cases}
    f(x_+) e^{-(x-x_+)\frac{\sqrt{a-1}}{x_+} }& x > x_+, \\
    f(a-1)-f(x_+) & a-1 < x \le x_+ ,\\
    f(a-1)-f(x_-) & x_- \le x \le a-1 , \\
    f(x_+) e^{-(x_-x)\frac{\sqrt{a-1}}{x_-} }& 0 < x < x_-.   
\end{cases}
$$
If we were to use these bounds in the modified acceptance-complement method, then we need to check that $\int q \le 1$.
Assuming $x_- > 0$, the area under the overall envelope is
\begin{align}
\label{A}
A
&\isdef 
\sqrt{a-1}  \times
\left(
 2f(a-1) - f(x_-) -f(x_+) + \frac {f(x_-) x_- + f(x_+) x_+}{a-1}
\right)\nonumber\\
&=
\sqrt{a-1} \times
\left(
 2f(a-1)  + \frac {f(x_+)- f(x_-) }{\sqrt{a-1}}
\right)\nonumber\\
&=
 2f(a-1)\sqrt{a-1}  + f(x_+)- f(x_-)\nonumber\\
&= f(a-1) \left( 2 \sqrt{a-1} + \frac{f(x_+)}{f(a-1)} -  \frac{f(x_-)}{f(a-1)} \right)\nonumber\\
&= \frac{1}{\Gamma (a)} \left( \frac{a-1}{e}\right)^{a-1} \left( 2 \sqrt{a-1} + e^{a-1-x_+} \left( \frac{x_+}{a-1} \right)^{a-1} - e^{a-1-x_-} \left( \frac{x_-}{a-1} \right)^{a-1} \right)\nonumber \\
&= \frac{1}{\Gamma (a)} \left( \frac{a-1}{e}\right)^{a-1} 
\left( 2 \sqrt{a-1} + \psi (\sqrt{a-1}) \right) , 
\end{align}
where
\begin{align}\label{psi}
    \psi (x) = 
e^{-x} \left( 1+ \frac{1}{x} \right)^{x^2} - e^{x} \left( 1- \frac{1}{x}  \right)^{x^2} .
\end{align}
Using a standard Stirling bound for the gamma function \cite{karatsuba2001asymptotic, mortici2011improved, mortici2011ramanujanlarge, mortici2011ramanujan, robbins1955remark}, we have
$$
f(a-1) 
= \frac{1} {\Gamma (a)} \left( \frac{a-1}{e} \right)^{a-1}
\le
\frac{1}{ \sqrt{2\pi (a-1)} ( 1 + 1/(2a-1))^{1/6}}.
$$
Combining this, one can show that $A < 1$ for $a \ge 5$.
To generate a random variate from $q$, we first select a random $J \in \{ 1,2,3,4\}$ to indicate which of the four parts of $q$ we should pick.  Clearly, $\PROB \{ J = i \}$ is proportional to the area under the bounding curve.  So, we have
$\PROB \{ J=i\} = p_i /(p_1+p_2+p_3+p_4)$, where
$$
p_i =
\begin{cases}
  \frac{ f(x_+) x_+ }{ \sqrt{a-1} } & i=1, \\
  (f(a-1)-f(x_+)) \sqrt{a-1} & i=2, \\
  (f(a-1)-f(x_-)) \sqrt{a-1} & i=3, \\
  \frac{ f(x_-) x_- }{ \sqrt{a-1} } & i=4.
\end{cases}
$$
A random variate $Y$ with density proportional to $q$ is then generated as follows, denoting a uniform $[0,1]$ random variable by $U$:
$$
Y =
\begin{cases}
  a-1 + \sqrt{a-1}\left(1 + \log\left(\frac 1U\right) \left(1 + \frac{1}{\sqrt{a-1}}\right) \right) & J=1, \\
  a-1 + U \sqrt{a-1} & J=2, \\
  a-1 - U \sqrt{a-1} & J=3, \\
  a-1 - \sqrt{a-1}\left(1 + \log\left(\frac 1U\right) \left(1 - \frac{1}{\sqrt{a-1}}\right) \right) & J=4.
\end{cases}
$$
We can now summarize the algorithm for the gamma law with parameter $a > 4.896$. To understand the algorithm, it is convenient to define
$$
g(x) = \frac{f(x)}{f(a-1)} = \left(\frac{x}{a-1}\right)^{a-1} e^{a-1-x},
$$
and
note that we can harmlessly replace $f$ by $g$ in the definition of the $p_i$'s.

\begin{algorithm}[H]
\caption{Generator for $r$ when $a \ge 2$}\label{r}
\begin{algorithmic}[4]
\State let $x_+ = a-1 + \sqrt{a-1}$, $x_- = a-1 - \sqrt{a-1}$
\State let $U, V$ be uniform on $[0,1]$
\If {$V \le f(x_+)/(f(x_+)+f(x_-))$}
\State return $X = a-1 + U \sqrt{a-1}$ and halt.
\Else
\State return $X = a-1 - U \sqrt{a-1}$ and halt.
\EndIf
\Comment{$X$ has density proportional to $r$}
\end{algorithmic}
\end{algorithm}

\begin{algorithm}[H]
\caption{Generator for the gamma density when $a \ge 5$}\label{gamma-gen}
\begin{algorithmic}[5]
\State $x_+ \gets a-1 + \sqrt{a-1}$, $x_- \gets a-1 - \sqrt{a-1}$, $A \gets $\eqref{A},\\
$p_1 \gets g(x_+)x_+/\sqrt{a-1}$, 
$p_2 \gets (g(a-1)-g(x_+))\sqrt{a-1}$ \\
$p_3 \gets (g(a-1)-g(x_-))\sqrt{a-1}$, 
$p_4 \gets g(x_-)x_-/\sqrt{a-1}$ \\
$p \gets p_1+p_2+p_3+p_4$
\State let $W$ be uniform on $[0,1]$
\If {$W \ge 1-A$}
\State let $J=j$ with probability $p_j/p$, $1 \le i \le 4$
\State let $U, V $ be uniform on $[0,1]$
\Switch{$J$}
        \Case{$J=1$}
            \State $Y \gets a-1 + \sqrt{a-1}\left(1 + \log\left(\frac 1U\right) \left(1 + \frac{1}{\sqrt{a-1}}\right) \right)$
            \If {$V f(x_+) e^{-(Y-x_+)\frac{\sqrt{a-1}}{x_+} } \le f(Y)$} \State return $X \gets Y$ and halt 
            \EndIf
        \EndCase
        \Case{$J=2$}
            \State $Y \gets a-1 + U \sqrt{a-1}$
            \If {$V (f(a-1)-f(x_+)) \le f(Y)-f(x_+)$} \State return $X \gets Y$ and halt
            \EndIf
        \EndCase
        \Case{$J=3$}
            \State $Y \gets a-1 - U \sqrt{a-1}$
            \If {$V (f(a-1)-f(x_-)) \le f(Y)-f(x_-)$} \State return $X \gets Y$ and halt
            \EndIf
        \EndCase
        \Case{$J=4$}
            \State $Y \gets a-1 - \sqrt{a-1}\left(1 + \log\left(\frac 1U\right) \left(1 - \frac{1}{\sqrt{a-1}}\right) \right)$
            \If {$V f(x_-) e^{-(x_- -Y)\frac{\sqrt{a-1}}{x_-} } \le f(Y)$} \State return $X \gets Y$ and halt 
            \EndIf
        \EndCase
\EndSwitch 
\EndIf
\State return $X$ with density proportional to $r$
\Comment{$X$ has density $f$}
\end{algorithmic}
\end{algorithm}

In this algorithm, all instances of $f$ can be replaced by $g$, so that we can avoid computing the gamma function. Unfortunately, the value
of $A$ depends upon $\Gamma (a)$.
If this value is available, then the above algorithm provides a constant-time method for simulating the gamma density. When $\Gamma (a)$ is not available, then the algorithm above needs to be adapted: see the next section.

Finally, when $a < 5$, then we can, perhaps repeatedly,  employ the observation that
$$
G_a \inlaw U^{\frac{1}{a}} G_{a+1},
$$
where $U$ is uniform $[0,1]$ and $G_a$ is a gamma $(a)$ random variable.  This would yield a one-liner for all gamma random variates for all $a > 0$.

\begin{remark}
\textsc{gamma random variates.}
For uniformly fast gamma random variates, we refer to the surveys
in \cite{D5} and \cite{Lu}. Among rejection algorithms,
the method of Marsaglia and Tsang \cite{Marsaglia+Tsang}
is highly recommended. Many simulation studies confirm
that the method of Schmeiser and Lal \cite{SL} is quite
competitive if the gamma parameter is at least one. Xi, Tan and Liu
\cite{Xi} suggested generating $\log (G_a)$ instead. As $\log (G_a)$
has a log-concave density for all values of $a > 0$, a uniformly
fast generator is quite easily obtained either by the universal
method of \cite{D4} or a specialized algorithm as developed, e.g.,
in \cite{D8}.
$\square$
\end{remark}

\begin{remark}
\textsc{derived distributions.}
We can combine one-liners to obtain additional ones. Examples include
\begin{itemize}
    \item[(i)] A beta $(a,b)$ random variate$B_{a,b}$ can be obtained as $G_a/(G_a + G_b)$, where $a,b > 0$.
    \item[(ii)] The beta prime distribution with parameters $a,b > 0$ is that of $B_{a,b}/(1-B_{a,b}) = G_a/G_b$.
    \item[(iii)] Teichroew's (or the variance-gamma) distribution, which is the law of $N\sqrt{G_a}$, where $N$ is a standard normal random variable \cite{teichroew1957mixture}.
\end{itemize}
\end{remark}


\section{The gamma density without access to the gamma function}\label{GammaWithout}

To take care of the remaining problem with the presence of $\Gamma (a)$ in \eqref{A}, we use a technique first suggested by Greaves \cite{greaves2026extended}.
In algorithm \ref{gamma-gen}, we find the condition
\begin{align}
\label{decision1}
\text{if } W \le 1-A \textrm{ then}\ldots ,
\end{align}
where $A \le 1$ is given by \eqref{A}, and $V$ is uniform on $[0,1]$ and independent of everything else.
Assume now that we can write
$$
A = \EXP \{ \varphi (Z) \},
$$
where $Z$ is a random variable, and $0 \le \varphi \le 1$.
Let us replace the decision \eqref{decision1} by 
\begin{align}
\label{decision2}
\text{if } V \le Z \textrm{ then}\ldots .
\end{align}
Note that both decisions are equivalent since
$$
\PROB \{ W \le 1-\varphi (Z)  \}
= \EXP \{ 1-\varphi (Z)  \}
= A 
= \PROB \{ W \le 1-A \}.
$$
We recall that
\begin{align*}
  A = \frac{1}{\Gamma (a)} \left( \frac{a-1}{e}\right)^{a-1} 
\left( 2 \sqrt{a-1} + \psi (\sqrt{a-1}) \right).  
\end{align*}
We propose $Z = |N|/\sqrt{a-1}$, where $N$ is a standard normal random variate, itself a function of two independent uniform $[0,1]$ random variates by the Box-M\"uller method. The density of $Z$ is
$$
\sqrt{2(a-1) \over \pi} e^{-\frac{(a-1)x^2}{2}}, x > 0.
$$
We define
$$
\varphi(x) = 
\frac{\left( 2 \sqrt{a-1} + \psi (\sqrt{a-1}) \right)}{\sqrt{2(a-1)\pi}}
e^{-(a-1)(\Phi(x)-x^2/2)} \I_{0 \le x\le \pi},
$$
where 
\begin{align}\label{Phi}
 \Phi(x) = 1 - \frac{x}{\tan x} + \log \left( \frac{x } {\sin x} \right).   
\end{align}
Greaves \cite{greaves2026} showed that $\Phi(x) - x^2/2$ is increasing on $[0,\pi)$, taking the value $0$ at $x=0$. Furthermore,
\begin{align*}
\EXP \{ \varphi (Z) \}
&= 
\frac{\left( 2 \sqrt{a-1} + \psi (\sqrt{a-1}) \right)}{\sqrt{2(a-1)\pi}}
\int_0^\pi \sqrt{2(a-1) \over \pi} e^{-\frac{(a-1)x^2}{2}} 
 e^{-(a-1)(\Phi(x)-x^2/2)} \, dx\\
& =  
 \left( 2 \sqrt{a-1} + \psi (\sqrt{a-1}) \right)
\frac{1}{\pi}\int_0^\pi  
 e^{-(a-1) \Phi(x)} \, dx\\
& =  
 \left( 2 \sqrt{a-1} + \psi (\sqrt{a-1}) \right)
 \frac{1}{\Gamma(a)}
 \left( \frac{a-1}{e} \right)^{a-1}
 \\
 &= A,
\end{align*}
where we used a Hankel contour representation
of the reciprocal gamma function \cite{temme1996special}.
We note that
\begin{align*}
0 
\le \varphi(Z) 
\le 
\frac{2 \sqrt{a-1} + \psi (\sqrt{a-1}) }{\sqrt{2(a-1)\pi}}
=
\sqrt{\frac{2}{\pi}} + \frac{\psi (\sqrt{a-1})}{\sqrt{2(a-1)\pi}}
< 1
\end{align*}
for $a \ge 6.2829$.
Therefore, we can avoid the gamma function altogether in algorithm \ref{gamma-gen} if $a \ge 6.2829$ and if the line ``if $W \ge 1-A$ then'' is replaced by these two lines,
\begin{align*}
&\textrm{let $Z=|N|/\sqrt{a-1}$ where $N$ is a standard normal random variable}\\
&\textrm{if } W \ge 1 - \frac{\left( 2 \sqrt{a-1} + \psi (\sqrt{a-1}) \right)}{\sqrt{2(a-1)\pi}}
e^{-(a-1)\Phi(Z)+N^2/2} \textrm { and } Z \le \pi \textrm { then  },
\end{align*}
where $\psi$ and $\Phi$ are defined in \eqref{psi} and \eqref{Phi}, respectively.

\begin{remark}
\textsc{open problem \# 1.}
The loopless algorithm given above---just like Greaves's algorithm---makes heavy use of branching and/or switching (i.e., ``if'' or ``case'' statements). It remains open whether a loopless algorithm without branching exists for the gamma distribution. 
$\square$
\end{remark}


\section{A universal acceptance-complement method for log-concave densities}\label{unimodal}

Assume that the density $f$ is log-concave with a mode at $m$.  
For this family, thanks to the inequality
\begin{align}\label{exp}
f(x) \le f(m) \min \left( 1 , e^{1 - |x-m| f(m)} \right),
\end{align}
there is a simple rejection algorithm with an expected number of iterations not exceeding $4$ \cite{D4}. That algorithm only needs access to a black box for $f$; higher derivatives of $f$ are not needed. Additional references on random variate generation for 
log-concave laws include \cite{HLD, LH, LH2, D5}.

The objective of this section is to show that there is a universal acceptance-complement method for this family as well---an ``off-the-shelf'' algorithm that applies to all log-concave densities. 

\begin{theorem}\label{theorem}
    Given the value of a mode $m$ and access to a black box that computes the value of the density $f$ at any point, there exists an implementation of the acceptance-complement method that applies to all log-concave densities on $\RR$ and has uniformly bounded execution time.
\end{theorem}

\begin{proof}
We propose a grid to define $r$ and $q$ based on the unimodality of $f$.
To this end, we need a good tail bound on the density $f$, sharper than \eqref{exp}.
Note that the largest value of $\log (f(x))$, $x>m$ occurs when $\log f$ is linear between $(m, \log(f(m))$ and $(x,\log (f(x)))$ and zero outside $[m,x]$. Assume that $(x-m)f(m)> 1$. The largest value of $f(x)$ is fixed by making the integral one, i.e.,
\begin{align*}
\int_0^x &f(m) e^{-(y-m) \frac{\log(f(m)/f(x))}{x-m}} \, dx\\
&= (x-m)f(m) \frac{1-f(x)/f(m)}{\log(f(m)/f(x))} \\
&= (x-m) \frac{f(m)-f(x)}{\log(f(m)/f(x))} \\
&= 1\\
&\ge (x-m) \frac{f(m)-f(m)e^{1-(x-m)f(m)}}{\log(f(m)/f(x))} \textrm{ (by \eqref{exp})}
\end{align*}
Thus,
\begin{align}\label{exp2}
f(x) \le f(m) \exp\left( -(x-m)f(m)\left(1-e^{1-(x-m)f(m)}\right)\right).
\end{align}
To this end, fix an integer $n>1$ and a grid width $\Delta = \delta / f(m)$ for $n\delta > 1$.  The line is partitioned into $2n+2$ intervals defined by the points $s_i = m + i \Delta$, $-n \le i \le n$.
We define positive functions $r$ and $Q$ by
$$
r(x) = f(s_i) \le f(x) \le f(s_{i-1}) = Q(x), s_{i-1} < x \le s_i, 1 \le i \le n, 
$$
$$
r(x) = f(s_i) \le f(x) \le f(s_{i+1}) = Q(x), s_{i} < x \le s_{i+1}, -n \le i \le -1, 
$$
$$
r(x) = 0 \le f(x) \le f(s_n) \exp\left(-(x-s_n) \frac{\log(f(m)/f(s_n))}{s_n-m} \right) = Q(x), x > s_n,
$$
and
$$
r(x) = 0 \le f(x) \le f(s_{-n}) \exp\left( -(s_{-n} -x) \frac{\log(f(m)/f(s_{-n}))}{m-s_{-n}} \right) = Q(x), x < s_{-n}.
$$
From \eqref{exp2}, we deduce
\begin{align}\label{exp3}
\max( f(s_n), f(s_{-n}))
\le f(m) \exp\left( -n\delta\left(1-e^{1-n\delta}\right)\right).
\end{align}
Set $q=Q-r$. With such a choice, we have
\begin{align*}
\int q 
&= 2\Delta f(m) -\Delta f(s_n) -\Delta f(s_{-n}) 
+ \int_{s_n}^\infty f(s_n) e^{-(x-s_n) \frac{\log(f(m)/f(s_n))}{s_n-m} } \,dx \\
&\qquad + \int_{-\infty}^{s_{-n}} f(s_{-n}) e^{-(s_{-n}-x) \frac{\log(f(m)/f(s_{-n}))}{m-s_{-n}} } \,dx\\
&= 2\Delta f(m) -\Delta f(s_n) -\Delta f(s_{-n})
+ \frac{n\Delta f(s_n)}{\log(f(m)/f(s_n))}
+ \frac{n\Delta f(s_{-n})}{\log(f(m)/f(s_{-n}))}\\
&\le 2\delta \left(1 + \exp\left( -n\delta\left(1-e^{1-n\delta}\right)\right) \left( \frac{ n}{n\delta\left(1-e^{1-n\delta}\right)} - 1 \right) \right).
\end{align*}
The first $n$ for which this bound drops below $1$ for some $x$ is $n=7$. Any choice
$$
0.3467 \le \delta \le 0.4643
$$
will do. The minimal value is approximately $0.9540$ and occurs for $\delta \approx 0.3968$.  We will take $\delta=2/5$ in the algorithm proposed below.
\end{proof}
The proof reveals that a partition of $\RR$ into $16$ intervals suffices for the acceptance-com{-}plement method to work.  Remarkably, the locations of these intervals can be fixed as a function of $f(m)$ and $m$ without knowing the shape of $f$ within the family of log-concave densities. For particular subfamilies, one can do with fewer intervals: for example, as we showed, four intervals suffice for the family of gamma densities of parameter $\ge 5$. As $f$ needs to be computed at $16$ points per generated random variate, this particular version of the acceptance-complement method cannot be competitive with respect to the universal rejection method for log-concave densities \cite{D4}, which requires on average not more than $5$ evaluations per generated variate.  Yet, for batch simulation, with the table of $15$ $f(s_i)$ values ($|i| \le 7$) pre-computed, the acceptance-complement method should be competitive.  Also, there is virtually nothing gained by increasing the size of the partition beyond $16$. We believe that this is the first simulation method for all log-concave densities with deterministically and uniformly bounded execution time.

The algorithms for $r$ and $q$ start with the selection of a discrete random variable taking only 16 values.  For this, we recommend Walker's method \cite{walker1974new, walker1977efficient},
which takes the integer part of a uniform random variable on $[0,16]$ to make one selection and uses the remainder to decide whether that selection should be changed.  

\begin{remark}\textsc{Open problem \# 2.}
 Our algorithm uses a partition into 16 intervals. What is the smallest partition size that works for all log-concave densities?  \end{remark}

We conclude this section with an implementation suggested by the proof of Theorem \ref{theorem}. To save space, we break the solution into three parts: the precomputation of table values, the algorithm for generating a random variate with a density proportional to the lower bound $r$, and the overall algorithm.

The values pre-computed include $M=f(m)$, $\delta = 2/5$, $n=7$, $\Delta = \delta/M$, $s_i = m+i\Delta$ ($-n \le i \le n$), and $f_i = f(s_i)$ ($-n \le i \le n$).
For the lower bound, we need $\int r$ over all individual intervals. To the right of $m$, these intervals are numbered from $1$ to $n$.  To the left of $m$, they are numbered from $-1$ down to $-n$. These integrals are
$$
r_i =
  f_i \Delta ,  -n \le i \le n, i\not=0 .
$$
Let the cumulative values be $R_i = \sum_{j=-n}^{i} r_j$.
For the upper bound $q \ge f-r$, we need $\int Q$ over all individual intervals. To the right of $m$, these intervals are numbered from $1$ to $n+1$, the last one being half-infinite.  To the left of $m$, they are numbered from $-1$ down to $-(n+1)$. These integrals are
$$
Q_i =
\begin{cases}
  f_{n} \frac{n\Delta }{ \log(M/f_n)}  & i=n+1 ; \\
  f_{i-1} \Delta  & 1 \le i \le n ; \\
  f_{i+1} \Delta  & -n \le i \le -1 ; \\
  f_{-n} \frac{n\Delta }{\log(M/f_{-n})}   & i=-(n+1) . 
\end{cases}
$$
We also need 
$$
A= \int q = \int Q -\int r= 2\delta - r_n -r_{-n} + Q_{n+1} + Q_{-(n+1)}. 
$$

\begin{algorithm}[H]
\caption{Generator for $r$}\label{r2}
\begin{algorithmic}[6]
\State generate $I$ such that $\PROB \{ J=j\} = r_j/R_n$, $-n \le j \le n, j\not=0. $
\State generate $U$ uniformly on $[0,1]$
\If {$I>0$} 
\State return $X \gets m+ (J- U)\Delta$
\Else
\State return $X \gets m - (J - U)\Delta$
\EndIf
\Comment{$X$ has density proportional to $r$}
\end{algorithmic}
\end{algorithm}

\begin{algorithm}[H]
\caption{Generator for a log-concave density $f$ with mode $m$}\label{logconcave}
\begin{algorithmic}[7]
\State let $W$ be uniform on $[0,1]$
\If {$W \ge 1-A$}
\State let $J=j$ with probability $(Q_j-r_j)/A$, $|j|\le n+1, j\not=0$ (where $r_{n+1}=r_{-(n+1)}=0$)
\State let $U, V $ be uniform on $[0,1]$
\Switch{$J$}
        \Case{$1 \le J \le n$}
        \State $Y \gets m+(J-U)\Delta$
        \If {$V (f_{J-1} - f_J) \le f(Y)-f_J$}
            \State return $X \gets Y$ and halt
        \EndIf
        \EndCase
        \Case{$-n \le J \le -1$}
        \State $Y \gets m-(J-U)\Delta$
        \If {$V (f_{J+1} - f_J) \le f(Y)-f_J$}
        \State  return $X \gets Y$ and halt
        \EndIf
        \EndCase
        \Case{$J=n+1$}
        \State $Y \gets m + n\Delta  + n\Delta \log\left(\frac 1U\right) \left( \frac{1}{\log (M / f_n)} \right) $
            \If {$VU f_n  \le f(Y)$} \State return $X \gets Y$ and halt 
            \EndIf
        \EndCase
        \Case{$J=-(n+1)$}
        \State $Y \gets m - n\Delta  - n\Delta \log\left(\frac 1U\right) \left( \frac{1}{\log (M / f_{-n})} \right) $
            \If {$VU f_{-n}  \le f(Y)$} \State return $X \gets Y$ and halt 
            \EndIf
        \EndCase
\EndSwitch 
\EndIf
\State return $X$ with density proportional to $r$
\Comment{$X$ has density $f$}
\end{algorithmic}
\end{algorithm}

\begin{remark}
\textsc{Table-based rejection algorithms for log-concave distributions.}
Gilks \cite{G1}, Gilks and Wild \cite{GW1, GW2} and Gilks, Best and Tan \cite{GBT} have suggested partitions of the sample space into intervals for designing fast rejection sampling algorithms for log-concave distributions.     
\end{remark}

\section{Extensions}\label{Extensions}

The previous example can be generalized quite easily to all (black box) unimodal densities for which one has an explicit upper bound on the densities in the left and right tails and knows the position of the mode.  These include unimodal distributions with known $r$-th central moment with $r > 0$ or with known finite value of $\int e^{\alpha |x|} f(x) \, dx$ for fixed $\alpha > 0$. One could also deal with multimodal densities provided one knows the locations of all maxima and minima, and has tail bounds as discussed above. In all these cases, there exists a finite partition that allows one to design an acceptance-complement-based random variate generator.  In some cases, the partitions could become quite large, and therefore, this approach is not competitive with rejection algorithms unless one desires to generate random variates in batches.


\section{Appendix: Proof of the validity of the acceptance-complement method}\label{ACProof}

Let $B$ be a Borel set.
Referring to algorithm \ref{AC1}, we have
\begin{align*}
\PROB &\{ X \in B \}\\
&=
\PROB \{ Y \in B , U q(Y) \le f(Y)-r(Y) \}
+
\PROB \{ X \in B , U q(Y) \ge f(Y)-r(Y) \}\\
&=
\int_B q(y) \frac{f(y)-r(y)}{q(y)}\, dy
+
\frac{\int_B r (x)\, dx} {\int r(x) \, dx} \times
\left(
1 - \int q(y) \frac{f(y)-r(y)}{q(y)}\, dy
\right)\\
&=
\int_B (f(y)-r(y))\, dy
+
\frac{\int_B r (x)\, dx} {\int r(x) \, dx} \times
\int r(y)\, dy\\
&=
\int_B f(y)\,dy.
\end{align*}


\bibliographystyle{plainnat}
\bibliography{p.bib}

@article{teichroew1957mixture,
  author    = {Teichroew, D.},
  title     = {The Mixture of Normal Distributions with Different Variances},
  journal   = {The Annals of Mathematical Statistics},
  volume    = {28},
  number    = {2},
  pages     = {510--512},
  year      = {1957},
  publisher = {Institute of Mathematical Statistics},
  doi       = {10.1214/aoms/1177706981},
  url       = {https://doi.org}
}

@article{walker1977efficient,
  title={An efficient method for generating discrete random variables with general distributions},
  author={Walker, A. J.},
  journal={ACM Transactions on Mathematical Software},
  volume={3},
  number={3},
  pages={253--256},
  year={1977},
  publisher={Association for Computing Machinery (ACM)},
  doi={10.1145/355744.355749},
  url={https://doi.org/10.1145/355744.355749}
}

@article{walker1974new,
  title={New fast method for generating discrete random numbers with arbitrary frequency distributions},
  author={Walker, A. J.},
  journal={Electronics Letters},
  volume={10},
  number={8},
  pages={127--128},
  year={1974},
  publisher={IET},
  doi={10.1049/el:19740097},
  url={https://doi.org}
}

@article{vonneumann1963various,
  author     = {von Neumann, John},
  title      = {Various techniques used in connection with random digits},
  journal   = {Collected Works},
  volume    = {5},
  pages     = {768--770},
  publisher = {Pergamon Press},
  year      = {1963},
  note      = {Also in Monte Carlo Method, National Bureau of Standards Series, vol.~12, pp.~36-38, 1951}
}

@misc{greaves2026,
  author        = {Greaves, Dylan},
  title         = {Extended One-Liners for the Gamma, Poisson, and Binomial
Distributions},
  year          = {2026}
}

@article{greaves2026extended,
  title={Extended One-Liners for the Beta, Gamma, and {D}irichlet Distributions with Shape Parameters Below One}, 
  author={Dylan Greaves},
  year={2026},
  volume={2604.11199},
  eprint={2604.11199},
  journal={arXiv},
  primaryClass={stat.CO}
}

@book{temme1996special,
  author        = {Temme, Nico M.},
  title         = {Special Functions: An Introduction to the Classical Functions of Mathematical Physics},
  publisher     = {Wiley},
  address       = {New York},
  year          = {1996}
}

@article{karatsuba2001asymptotic,
  title     = {On the asymptotic representation of the {E}uler gamma function by {R}amanujan},
  author    = {Karatsuba, Ekatherina A.},
  journal   = {Journal of Computational and Applied Mathematics},
  volume    = {135},
  number    = {2},
  pages     = {225--240},
  year      = {2001},
  publisher = {Elsevier},
  doi       = {10.1016/S0377-0427(00)00586-0},
  note      = {MR 1850542}
}

@article{mortici2011ramanujan,
  title     = {Ramanujan's estimate for the gamma function via monotonicity arguments},
  author    = {Mortici, Cristinel},
  journal   = {The Ramanujan Journal},
  volume    = {25},
  number    = {2},
  pages     = {149--154},
  year      = {2011},
  publisher = {Springer},
  doi       = {10.1007/s11139-010-9265-y}
}

@article{mortici2011improved,
  title     = {Improved asymptotic formulas for the gamma function},
  author    = {Mortici, Cristinel},
  journal   = {Computers \& Mathematics with Applications},
  volume    = {61},
  number    = {11},
  pages     = {3364--3369},
  year      = {2011},
  publisher = {Elsevier},
  doi       = {10.1016/j.camwa.2011.04.036}
}

@article{mortici2011ramanujanlarge,
  title     = {On {R}amanujan's large argument formula for the gamma function},
  author    = {Mortici, Cristinel},
  journal   = {The Ramanujan Journal},
  volume    = {26},
  number    = {2},
  pages     = {185--192},
  year      = {2011},
  publisher = {Springer},
  doi       = {10.1007/s11139-010-9281-y}
}

@article{robbins1955remark,
  author    = {Robbins, Herbert},
  title     = {A Remark on {S}tirling's Formula},
  journal   = {The American Mathematical Monthly},
  volume    = {62},
  number    = {1},
  pages     = {26--29},
  year      = {1955},
  publisher = {Mathematical Association of America},
  doi       = {10.2307/2308012}
}

@article{Lu,
author = "Luengo, E.A.",
title = "Gamma pseudo-random number generators",
journal = "ACM Computing Surveys",
volume = "55(4)",
pages = "85",
year = "2022"
}

@article{Xi,
author = "Xi, B.  and Tan, K.M.  and Liu, C.",
title = "Logarithmic transformation-based gamma random number generators",
journal = "Journal of Statistical Software",
volume = "55(4)",
pages = "1--17",
year = "2013"
}

@article{SL,
author = "Schmeiser, B.  and Lal, R.",
title = "Squeeze methods for generating gamma variates",
journal = "Journal of the American Statistical Association",
volume = "75",
pages = "679--682",
year = "1980"
}

@article{ahrens1982computer,
  author    = {Ahrens, J. H. and Dieter, U.},
  title     = {Computer Generation of {P}oisson Deviates from Modified Normal Distributions},
  journal   = {ACM Transactions on Mathematical Software},
  volume    = {8},
  number    = {2},
  pages     = {163--170},
  year      = {1982},
  publisher = {ACM New York, NY, USA},
  doi       = {10.1145/355993.355997}
}

@article{deak1981economical,
  author    = {Deak, I.},
  title     = {An Economical Method for Random Number Generation and a Normal Generator},
  journal   = {Computing},
  volume    = {27},
  number    = {2},
  pages     = {113--121},
  year      = {1981},
  publisher = {Springer},
  doi       = {10.1007/BF02243545}
}

@article{Kronmal1,
author = "Kronmal, Richard A. and Peterson, Jr. Arthur V.",
title = "A variant of the acceptance-rejection method for the computer generatlon of random variables",
journal = "Journal of the American Statistical Association",
volume = "76",
pages = "446-451",
year = "1981"
}

@article{Kronmal2,
author = "Kronmal, Richard A. and Peterson, Jr. Arthur V.",
title = "An acceptance-complement analogue of the mixture-plus-acceptance-rejection method for generating random variables",
journal = "ACM Transactions on Mathematical Software",
volume = "10",
number = "3",
pages = "271-281",
year = "1984"
}

@article{Bailey,
author = "Bailey, R.W.",
title = "Polar generation of random variates with the $t$ distribution",
journal = "Mathematics of Computation",
volume = "62",
pages = "779--781",
year = "1994"
}

@inproceedings{D1,
author = "Devroye, L.",
title = "Random variate generation in one line of code",
booktitle = "1996 Winter Simulation Conference Proceedings",
editor = "Charnes, J.M. and Morrice, D.J. and Brunner, D.T.  and Swain, J.J.",
publisher = "ACM",
pages = "265--272",
address = "San Diego, CA",
year = "1996"
}

@article{U,
author = "Ulrich, G.",
title = "Computer generation of distributions on the m-sphere",
journal = "Applied Statistics",
volume = "33",
pages = "158--163",
year = "1984"
}

@article{D8,
author = "Devroye, L.",
journal = "Statistics and Computing",
title = "Random variate generation for the generalized inverse {G}aussian distribution",
volume = "24",
pages = "239--246",
year = "2014"
}

@article{D4,
author = "Devroye, L.",
journal = "Computing",
title = "A simple algorithm for generating random variates with a log-concave density",
volume = "33",
pages = "247--257",
year = "1984"
}

@book{D5,
author = "Devroye, L.",
title = "Non-Uniform Random Variate Generation",
publisher = "Springer-Verlag",
address = "New York",
year = "1986"
}

@incollection{G1,
author = "Gilks, W.R.",
title = "Derivative-free adaptive rejection sampling for {G}ibbs sampling",
booktitle = "Bayesian Statistics 4",
editor = "Bernardo, J. and Berger, J. and Dawid, A.P. and Smith, A.F.M.",
publisher = "Oxford University Press",
year = "1992"
}

@article{GBT,
author = "Gilks, W.R.  and Best, N.G.  and Tan, K.K.C.",
title = "Adaptive rejection {M}etropolis sampling",
journal = "Applied Statistics",
volume = "44",
pages = "455--472",
year = "1995"
}

@article{GW1,
author = "Gilks, W.R.  and Wild, P.",
title = "Adaptive rejection sampling for {G}ibbs sampling",
journal = "Applied Statistics",
volume = "41",
pages = "337--148",
year = "1992"
}

@article{GW2,
author = "Gilks, W.R.  and Wild, P.",
title = "Algorithm AS 287: Adaptive Rejection Sampling from Log-Concave density function",
journal = "Applied Statistics",
volume = "41",
pages = "701--709",
year = "1993"
}

@book{HLD,
title = "Automatic Nonuniform Random Variate Generation",
author = "H{\"o}rmann, W.  and Leydold, J.  and Derflinger, G.",
publisher = "Springer-Verlag",
address = "Berlin",
year = "2004"
}

@inproceedings{LH,
author = "Leydold, J.  and H{\"o}rmann, W.",
title = "Black box algorithms for generating non-uniform continuous random variates",
editor = "Jansen, W. and Bethlehem, J.G.",
booktitle = "COMPSTAT 2000",
pages = "53--54",
year = "2000"
}

@article{LL,
author = "Jones, M.C.",
title = "Student’s simplest distribution",
journal = "Journal of the Royal Statistical Society Series D",
volume = "51",
pages = "41--49",
year = "2002"
}

@incollection{LH2,
author = "Leydold, J.  and H{\"o}rmann, W.",
title = "Universal algorithms as an alternative for generating non-uniform continuous random variates",
editor = "Schuler, G.I. and Spanos, P.D.",
booktitle = "Monte Carlo Simulation",
pages = "177--183",
year = "2001"
}

@article{Marsaglia+Tsang,
  title={A simple method for generating gamma variables},
  author={George Marsaglia and Wai Wan Tsang},
  journal={ACM Transactions on Mathematical Software},
  volume={26},
  number={3},
  pages={363--372},
  year={2000},
  publisher={ACM}
}

@article{chambers1976method,
  author    = {Chambers, J. M. and Mallows, C. L. and Stuck, B. W.},
  title     = {A method for simulating stable random variables},
  journal   = {Journal of the American Statistical Association},
  volume    = {71},
  pages     = {340--344},
  year      = {1976}
}

@article{kanter1975stable,
  author    = {Kanter, M.},
  title     = {Stable densities under change of scale and total variation inequalities},
  journal   = {Annals of Probability},
  volume    = {3},
  pages     = {697--707},
  year      = {1975}
}

@article{zolotarev1966representation,
  author    = {Zolotarev, V. M.},
  title     = {On the representation of stable laws by integrals},
  journal   = {Selected Translations in Mathematical Statistics and Probability},
  volume    = {6},
  pages     = {84--88},
  year      = {1966}
}

@book{zolotarev1986one,
  author    = {Zolotarev, V. M.},
  title     = {One-Dimensional Stable Distributions},
  publisher = {American Mathematical Society},
  address   = {Providence, R.I.},
  year      = {1986}
}
\end{document}